\documentclass{article} 
\usepackage{iclr2026_conference,times}

\usepackage{amsmath,amsfonts,bm}

\def\eqref#1{equation~\ref{#1}}

\def\1{\bm{1}}

\DeclareMathAlphabet{\mathsfit}{\encodingdefault}{\sfdefault}{m}{sl}
\SetMathAlphabet{\mathsfit}{bold}{\encodingdefault}{\sfdefault}{bx}{n}

\usepackage{amsmath}
\usepackage{amssymb}
\usepackage{mathtools}
\usepackage{amsthm}
\usepackage{bm}
\usepackage{multirow}
\usepackage{hyperref}
\usepackage{url}
\usepackage{booktabs}       
\usepackage{amsfonts}       
\usepackage{nicefrac}       
\usepackage{microtype}      
\usepackage{xcolor}  
\usepackage{bbding}
\usepackage{wrapfig}

\theoremstyle{plain}
\newtheorem{theorem}{Theorem}[section]
\newtheorem{proposition}[theorem]{Proposition}

\newtheorem{definition}[theorem]{Definition}

\theoremstyle{remark}

\newcommand{\methodname}{{\sc ENS}}

\newcommand{\x}{\bm{x}}

\renewcommand{\k}{\bm{k}}

\newcommand{\sk}{\textup{\textsf{sk}}}

\renewcommand{\k}{\bm{k}}

\newcommand{\cF}{\mathcal{F}}

\newcommand{\cK}{\mathcal{K}}

\newcommand{\cP}{\mathcal{P}}

\newcommand{\bbE}{\mathbb{E}}

\newcommand{\ens}{\textup{\textrm{ENS}}}

\usepackage{algorithm}
\usepackage[noend]{algpseudocode}
\usepackage{bbding}

\title{An Ensemble Framework for Unbiased Language Model Watermarking}

\author{%
  Yihan Wu\thanks{Equal contribution}~, ~Ruibo Chen$^*$, Georgios Milis, ~Heng Huang \\
  Department of Computer Science\\
  University of Maryland, College Park\\
  \texttt{\{ywu42, rbchen, milis, heng\}}@umd.edu}

\iclrfinalcopy 
\begin{document}

\maketitle

\begin{abstract}
As large language models become increasingly capable and widely deployed, verifying the provenance of machine-generated content is critical to ensuring trust, safety, and accountability. Watermarking techniques have emerged as a promising solution by embedding imperceptible statistical signals into the generation process. Among them, unbiased watermarking is particularly attractive due to its theoretical guarantee of preserving the language model's output distribution, thereby avoiding degradation in fluency or detectability through distributional shifts. However, existing unbiased watermarking schemes often suffer from weak detection power and limited robustness, especially under short text lengths or distributional perturbations. In this work, we propose \methodname, a novel ensemble framework that enhances the detectability and robustness of logits-based unbiased watermarks while strictly preserving their unbiasedness. \methodname\ sequentially composes multiple independent watermark instances, each governed by a distinct key, to amplify the watermark signal. We theoretically prove that the ensemble construction remains unbiased in expectation and demonstrate how it improves the signal-to-noise ratio for statistical detectors. Empirical evaluations on multiple LLM families show that \methodname\ substantially reduces the number of tokens needed for reliable detection and increases resistance to smoothing and paraphrasing attacks without compromising generation quality. 
\end{abstract}

\section{Introduction}

The rapid progress of large language models has enabled them to generate human‑level text at scale, raising urgent questions about provenance, accountability, and misuse of AI‑generated content.  
A leading line of defence is \emph{watermarking}, which embeds hidden statistical signals into the generation process so that downstream detectors can later verify authorship with high confidence\,\citep{Aaronson2022,kirchenbauer2023watermark,christ2023undetectable,kuditipudi2023robust,hu2023unbiased,wu2023dipmark,chen2024mark,chen2024enhancing,liu2024adaptive,chen2025improved,mao2024watermark,dathathri2024scalable}.  
Ideally, a watermark should satisfy three properties:  
\emph{(i) imperceptibility}: it should leave fluency and semantics intact;  
\emph{(ii) detectability}: it should be reliably identified from a moderate amount of text; and  
\emph{(iii) robustness}: it should withstand natural corruptions and adversarial removal attempts.

Among these methods, a particularly important subclass is \emph{unbiased} (a.k.a.\ distortion-free) watermarking~\citep{Aaronson2022,christ2023undetectable,kuditipudi2023robust,hu2023unbiased,wu2023dipmark,chen2025improved,mao2024watermark,dathathri2024scalable}, which modifies the generation distribution in a way that preserves its expectation. In other words, the average output distribution of a watermarked model remains indistinguishable from the original model, ensuring that watermarking does not degrade text quality or introduce detectable artifacts. This property makes unbiased watermarking especially attractive for real-world deployment, where imperceptibility and content fidelity are critical.

However, unbiased watermarking faces key limitations in practice. Because the expected distribution is unchanged, the statistical signal available to the detector is inherently weak. This often necessitates longer token sequences for detection and leaves the watermark more vulnerable to attacks such as sampling smoothing, truncation, or paraphrasing. To address these challenges, we propose \methodname, an \emph{ensemble} framework that composes multiple unbiased watermark instances with independent keys to amplify the detection signal while provably preserving the LM distribution.  
Intuitively, each watermark layer makes a tiny, unbiased perturbation; stacking \(n\) such layers aggregates the bias under the detector’s statistic, boosting signal‑to‑noise ratio roughly by \(\sqrt{n}\) while keeping generation quality intact.

Our main contributions are summarized as follows:
\begin{itemize}
    \item We introduce \methodname, a general ensemble framework for unbiased watermarking. Our method sequentially applies multiple independent unbiased reweighting functions, using distinct keys, to enhance the detection signal without altering the expected distribution.

    \item We prove that \methodname\ remains unbiased in expectation under independent keys, and we analyze how the ensemble size affects the watermark’s detectability and robustness.

    \item We empirically demonstrate that \methodname\ significantly improves detection performance and robustness to common perturbations across multiple model families and watermarking baselines, while preserving generation quality.
\end{itemize}

    
    

\section{Related Work}
\citet{kirchenbauer2023watermark} refined the statistical watermarking framework initially introduced by \citet{Aaronson2022}. They divided the language model tokens into red and green lists and favored the green list tokens by adjusting their logits with a fixed increment $\delta$. However, the proposed approach can inevitably decrease the generation quality of the text. To maintain the original output distribution in watermarked content, researchers have investigated novel approaches for token distribution modification. There are generally two types of unbiased watermark: sampling based unbiased watermark and logits based unbiased watermark. For sampling based unbiased watermark, we use the pseudo-random token sampler to sample the next token. For logits based unbiased watermark, we adapt pseudo-random adjuster to modify the LM token logits, then sampling the next token based on the modified LM logits.

\noindent\textbf{Sampling based unbiased watermarks.}  \citet{Aaronson2022} pioneered an unbiased watermarking method using Gumbel-max to pseudo-randomly sample the next token with  prefix n-grams as watermark keys. \citet{christ2023undetectable} used inverse sampling as the watermark sampler on a binary language model with watermark keys based on token positioning. ITS-edit and EXP-edit \cite{kuditipudi2023robust} refined the inverse-sampling and Gumbel-max strategies with a predetermined watermark key list. \citet{hu2023unbiased} used inverse-sampling and design a model-based LLR score for detection.  STA-1 \cite{mao2024watermark} adapted rejected-sampling strategy to improve the quality of the watermarked text under the low-entropy scenarios. \citet{dathathri2024scalable} proposed SynthID, which used tournament sampling to achieved better detectability, however, their approach required sampling a redundant $2^{30}$ tokens for generating a single token, which significantly increase the computational costs. 

\noindent\textbf{Logits based unbiased watermarks.}  
\citet{hu2023unbiased} introduce the first logits based unbiased watermark, $\gamma$-reweight, for watermarking, though their detection method is not model-agnostic. DiPmark \cite{wu2023dipmark} enhanced the $\gamma$-reweight technique and introduced a model-agnostic detector. \citet{chen2025improved} introduced MCmark, which significantly improved the detectability of the unbiased watermark.

\section{Preliminaries}

\paragraph{Notation.} 
We follow the notation conventions from prior work~\cite{hu2023unbiased,wu2023dipmark,mao2024watermark,chen2025improved} to describe the language model generation process. Let \( V \) denote the vocabulary set with cardinality \( N = |V| \). Define \( \mathcal{V} \) as the set of all possible token sequences of any length (including the empty sequence). Given a prompt, the LM generates tokens autoregressively. The probability of generating the next token \( x_{t+1} \in V \), conditioned on the preceding sequence \( x_{1}, \ldots, x_{t} \), is denoted by \( P_M(x_{t+1} \mid \x_{1:t}) \in \mathcal{P} \), where \( \mathcal{P} \) represents the space of token distributions.

\paragraph{Watermark Generator.}
To embed watermarks, the service provider modifies the original LM distribution \( P_M \) using either sampling-based or logits-based watermarking methods. The watermarking process involves reweighting \( P_M(\cdot \mid \x_{1:t}) \) to a watermarked distribution \( P_{M,w}(\cdot \mid \x_{1:t}, k) \), where \( k \in \mathcal{K} \) is a private key sampled from a known key space distribution \( P_{\mathcal{K}}(k) \). In practice, this key is typically derived from a hash function \( h \) applied to a \emph{secret key} \( \mathsf{sk} \) and a context identifier (e.g., an \( n \)-gram~\cite{Aaronson2022} or token position~\cite{christ2023undetectable}).

\paragraph{Logits-Based Reweighting.}
In logits-based watermarking, the watermark generator applies a reweighting strategy \( F \in \mathcal{F}: \mathcal{P} \times \mathcal{K} \to \mathcal{P} \), which transforms the original LM distribution \( P_M(\cdot \mid x_{1:t}) \) into a watermarked distribution \( F(P_M(\cdot \mid x_{1:t}), k) \). Following~\cite{hu2023unbiased}, a reweighting function \( F \) is said to be \emph{unbiased} if it preserves the original distribution in expectation over the key space. Specifically, for any \( P_M \in \mathcal{P} \) and \( x_{t+1} \in V \),
\begin{equation}\label{eq:unbiased-F}
    \mathbb{E}_{k \sim P_{\mathcal{K}}} \left[ F(P_M(x_{t+1} \mid \x_{1:t}), k) \right] = P_M(x_{t+1} \mid \x_{1:t}).
\end{equation}

\paragraph{Watermark Detector.}
At inference time, the \emph{detector} receives:
(i)~the putative text sequence \(\x_{1:T}\),
(ii)~the watermark keys (or a means to regenerate them), and
(iii)~the public description of \(F\).
Detection is framed as a binary hypothesis testing problem: \( H_0 \): The sequence was generated without a watermark; \( H_1 \): The sequence was generated with a watermark.
The detector constructs a statistical score function based on \( k \) and \( F \), which exhibits different distributions under \( H_0 \) and \( H_1 \).
Under \(H_1\), the statistic is stochastically larger (or smaller) than under \(H_0\), enabling detection with controlled type‑I error.
This allows for distinguishing between watermarked and non-watermarked content using standard hypothesis testing techniques.

\section{Methodology}\label{sec:method}
In this section, we introduce an ensemble framework for logits-based unbiased watermarks, aiming to improve both detection reliability and robustness without sacrificing unbiasedness. 
We propose an \emph{ensemble} construction that composes any base logits-based unbiased watermarking rule $F$ with multiple independently drawn watermark keys. Intuitively, each application of $F$ injects an (unbiased) weak statistical signal aligned with its key; composing $n$ such transforms stacks $n$ weak signals. Because unbiasedness is defined in expectation over the key distribution, sequential composition preserves marginal unbiasedness to the underlying LM distribution when keys are independent. At detection time, we exploit the availability of all $n$ watermark keys to aggregate evidence across them, yielding improved statistical power at the same (or controlled) false positive rate (FPR). This section formalizes the construction, proves unbiasedness, develops score aggregation and significance testing procedures, and discusses practical design choices (key scheduling, context partitioning, computational cost, and robustness considerations).

\begin{definition}[Watermark ensemble] Given an original LM distribution $P_M(\cdot|\x_{1:t})\in\cP$, a logits-based reweight strategy $F\in\cF$, and $n$ watermark keys $\k_{1:n}$, we define the $n$-fold ensemble transform $\ens:\mathbb{N}\times\cP \times\cF\times\cK^n\to\cP$ recursively , 
\begin{equation}\label{eq:unbiased-ensemble}
    \ens(n,F,P_M(\cdot|\x_{1:t}),\k_{1:n})=\left\{
    \begin{aligned}
        &F(\ens(n-1,F,P_M(\cdot|\x_{1:t}),\k_{1:n-1}),k_n), &n>1;\\
        &F(P_M(\cdot|\x_{1:t}),k_1),  &n=1.
    \end{aligned}
    \right.
\end{equation}
For brevity we suppress the conditioning context $\x_{1:t}$ when unambiguous and write $\ens_n(P_M)$.
\end{definition}

The ensemble mechanism can be viewed as a multi-layer reweighting pipeline, where each layer applies a subtle perturbation governed by its corresponding key. The resulting distribution remains close to $P_M$ in total variation, but the watermark signal accumulates across layers, enhancing the detectability of the watermark. The detailed watermarking algorithm can be found in Alg.~\ref{alg:generator} and Alg.~\ref{alg:detector}.

\paragraph{Sequential vs. parallel views.} The recursive definition is \emph{sequential}: each new key reweights the distribution output by the previous step. An equivalent \emph{parallel} interpretation, which is useful for certain implementations, treats the final logits as a sum of per-key adjustments applied to the original logits before reweighting. Sequential and parallel forms are algebraically equivalent for common families of $F$ (e.g., additive logit shifts, multiplicative temperature scaling within greenlists).

\subsection{Unbiasedness of the Ensemble}\label{subsec:ensemble-unbiased}
A crucial property of unbiased watermarking is the preservation of the original distribution's expectation over random keys. We show that this property holds under our ensemble construction:
\begin{theorem}[Unbiasedness] \label{thm:ensemble-unbiased} If $F$ is an unbiased logits-based reweight strategy, and the watermark keys $\k_{1:n}$ are i.i.d. from $P_\cK$, then the $n$-ensemble of $F$ is also an unbiased logits-based reweight strategy, i.e., 
\begin{equation}\label{eq:unbiased-ensemble}
    \mathbb{E}_{\k_{1:n}\sim P_\cK^n}[\ens_n(P_M(\cdot|\x_{1:t}))]=P_{M}(\cdot \mid \x_{1:t})
\end{equation} holds for arbitrary $P_M(\cdot|\x_{1:t})\in\cP$,
\end{theorem}

\begin{proof}[Proof sketch] We argue by induction on $n$. The base case ($n{=}1$) reduces to Eq.~\ref{eq:unbiased-F}. Suppose Eq.~\ref{eq:unbiased-ensemble} holds for $n{-}1$. For $n$, condition on the first $n{-}1$ keys and apply the law of iterated expectation:
\begin{align}
\bbE_{\k_{1:n}}[\ens_n(P_M(\cdot|\x_{1:t}))(v)] &= \bbE_{\k_{1:n-1}}\Big[\bbE_{k_n}[F(\ens_{n-1}(P_M(\cdot|\x_{1:t}))\mid k_n)(v)]\Big]\\
&\stackrel{(a)}{=} \bbE_{\k_{1:n-1}}[\ens_{n-1}(P_M(\cdot|\x_{1:t}))(v)]\stackrel{(b)}{=} P_M(v|\x_{1:t}),
\end{align}
where $v$ is an arbitrary token in $V$, (a) applies Eq.~\ref{eq:unbiased-F} with $P=\ens_{n-1}(P)$ (valid because $F$ is unbiased \emph{for any input distribution}), and (b) uses the induction hypothesis. \qedhere
\end{proof}

\paragraph{Remarks.} (i) Independence is sufficient but not necessary: certain correlated key schedules also preserve unbiasedness if the marginal over the final key $k_n$ is uniform and $F$ obeys a linearity property. (ii) If $F$'s unbiasedness holds only approximately (e.g., due to numerical truncation or top-$K$ filtering), ensemble bias accumulates at most linearly in $n$.

\textbf{Watermark key design.} The unbiasedness of $\ens_n$ requires that the watermark keys be independent during \textit{a single generation step}. Recall that a watermark key $k$ is typically derived as $h(\sk, \text{n-gram})$, where $h$ is a hash function ensuring independence across different $(\sk, \text{n-gram})$ pairs.  
Following this principle, there are two ways to generate independent keys for ensemble watermarking: a) Use $n$ distinct hash functions $h_1, \dots, h_n$, so that in each generation step, $h_1(\sk,\text{n-gram}), \dots, h_n(\sk,\text{n-gram})$ are independent.  b) Use $n$ distinct secret keys $\sk_1, \dots, \sk_n$, so that $h(\sk_1,\text{n-gram}), \dots, h(\sk_n,\text{n-gram})$ remain independent. In our implementation, we use b) to ensure the independence of the watermark keys during generation.

\subsection{Detect Efficiency}
The detection of the ensemble watermark is straightforward.
If the ensemble is built upon a logit-based strategy $F$, we can simply apply the detection algorithm of $F$ $n$ times using the secret keys $(\sk_1, \dots, \sk_n)$, and then aggregate the resulting detection scores.

\paragraph{Detection with Ensemble Watermarks.}
Given a generated sequence \(\x_{1:T}\), a detector score function $S$, and secret keys $(\sk_1,...,\sk_n)$, the detector computes per‑key scores \(\{S(\x_{1:T},\sk_i)\}_{i=1}^n\) and aggregates them, e.g.
\[
S_{\ens}(\x_{1:T}) \;=\; \sum_{i=1}^n S(\x_{1:T},\sk_i).
\]
Because the watermark bias adds \emph{coherently} across independent keys, \(S_{\ens}\) enjoys a higher signal‑to‑noise ratio, enabling stronger hypothesis testing between \(H_0\) (unwatermarked) and \(H_1\) (watermarked).




\paragraph{Signal and Variance Behavior.}
Unbiasedness ensures that, a downstream observer without keys cannot distinguish $\ens_n(P)$ from $P$. Detection instead leverages \emph{conditional} shifts introduced by each key. Let $S(x_t,\sk_i)$ denote the per-token score computed by the detector when evaluated with secret key $\sk_i$. Under the alternative hypothesis $H_1$ (watermarked with the same keys), the statistic $S(x_t,\sk_i)$ has a mean shifted by a positive amount $\mu_i > 0$ relative to its mean under the null hypothesis $H_0$ (no watermark). The ensemble detector aggregates the shift across keys by combining either (a) raw per-key log-likelihood ratios, (b) standardized $z$-scores, or (c) non-parametric ranks.

Because generation uses \emph{all} $n$ keys, the expected per-token shift in $S(x_t,\sk_i)$ generally increases with $n$, but its variance also grows due to interactions among keys. For many logit-based reweighting $F$ of interest (greenlist reweighting~\citep{kirchenbauer2023watermark}; DiP-style permute-reweight \citep{wu2023dipmark}), the per-key signals add approximately linearly at small $n$, yielding a signal-to-noise ratio (SNR) that scales roughly as $\sqrt{n}$ for the aggregated statistic. We formalize one such regime below.

\begin{proposition}[Approximate SNR scaling]\label{prop:snr}
Assume (i) conditional independence of per-key centered scores given the underlying token, (ii) common per-key variance $\sigma^2$, and (iii) common mean shift $\mu$ under $H_1$. Then the sum statistic $S_{\ens}(\x_{1:T}) = \sum_{i=1}^n S(\x_{1:T},\sk_i)$ has mean $n\mu$ and variance $n\sigma^2$ under $H_1$, yielding $\mathrm{SNR}(S_{\ens})=\mu\sqrt{n}/\sigma$. Under $H_0$, $T$ has mean $0$ and variance $n\sigma_0^2$. Hence, for fixed FPR calibrated under $H_0$, power increases with $n$.
\end{proposition}

\paragraph{Example (DiPmark detector).}
We illustrate the improved detecting strength of ensemble watermarking using the DiPmark detector. In the DiPmark detector, the red–green lists for generating each token are reconstructed using the secret key and the corresponding n-gram. The detector then counts the total number of green tokens in the sequence and applies a statistical test to determine watermark presence.

For a given key $\sk_i$, let $V_G(\x_{1:T};\sk_i)$ denote the number of tokens in the sequence $\x_{1:T}$ that fall into the \emph{green} set induced by $\sk_i$. The DiPmark score is
\[
S_{\mathrm{DiP}}(\x_{1:T},\sk_i)\;=\;\frac{V_G(\x_{1:T};\sk_i)}{T}-0.5.
\]
Under the $H_0$ (unwatermarked text), the green indicators are i.i.d.\ Bernoulli$(1/2)$, so by Hoeffding's inequality the (one-sided) $p$-value is bounded by $p_{\text{single}}\;\le\;\exp\!\bigl(-2T\,S_{\mathrm{DiP}}(\x_{1:T},\sk_i)^2\bigr).$

Ensembling $n$ DiPmark detectors:
Given keys $\sk_1,\ldots,\sk_n$, define the ensemble score
\[
S_{\ens}(\x_{1:T})\;=\;\sum_{i=1}^n S_{\mathrm{DiP}}(\x_{1:T},\sk_i)
\;=\;\frac{\sum_{i=1}^n V_G(\x_{1:T};\sk_i)}{T}\;-\;\frac{n}{2}.
\]
Equivalently, writing $Z_{t,i}\in\{0,1\}$ for the green indicator of token $t$ under key $\sk_i$, we have $S_{\ens}(\x_{1:T}) \;=\; n\!\left(\frac{1}{nT}\sum_{i=1}^n\sum_{t=1}^T Z_{t,i}-\frac{1}{2}\right)$.
Assuming independence across $t$ and $i$ under the null, Hoeffding's inequality applied to the average over $nT$ bounded variables yields
\[
p_{\ens}\;=\;\Pr\!\bigl(S_{\ens}\ge s\bigr)\;\le\;\exp\!\left(-\,\frac{2T}{n}\,s^2\right)
\quad\Longrightarrow\quad
p_{\ens}\;\le\;\exp\!\left(-\,\frac{2T}{n}\,S_{\ens}(\x_{1:T})^2\right).
\]
In the special case where all single-key scores are equal, $S_{\mathrm{DiP}}(\x_{1:T},\sk_1)=\cdots=S_{\mathrm{DiP}}(\x_{1:T},\sk_n)=s_0$, we have $S_{\ens}=ns_0$ and hence
\[
p_{\ens}\;\le\;\exp\!\bigl(-2Tn\,s_0^{2}\bigr)
\;=\;\Bigl[\exp\!\bigl(-2T\,s_0^{2}\bigr)\Bigr]^{\!n}
\;\approx\;\bigl(p_{\text{single}}\bigr)^{n}.
\]
Thus, ensembling $n$ independent DiPmark detectors improves the exponent linearly in $n$, i.e., the $p$-value decays exponentially with $n$. In practice, the ensemble $p$-value need not decay exponentially in $n$, because the per-key DiPmark detectors can become less informative when used jointly (e.g., due to the attenuation of individual scores). Nevertheless, ensembling typically improves overall detection power (efficiency) and yields a smaller $p$-value than any single detector.

\paragraph{Dealing with Key Dependence.}
In practice keys may not be statistically independent across tokens because the keying function $h(\sk, c)$ often reuses a fixed secret $\sk$ with context $c$ derived from overlapping $n$-grams or positions. Overlapping contexts induce correlations in the green/boosted sets across ensemble members, which in turn correlate the per-key scores. To address repeated contexts \(c\), we follow \cite{hu2023unbiased} and maintain a history of previously seen contexts during watermark generation; if the current \(c\) already appears in the history, we bypass watermarking and sample from the original (unwatermarked) distribution.

\subsection{Effect of the ensemble size $n$ on detectability.}\label{sec:ensemble size}
Increasing the ensemble size $n$ can strengthen detection by aggregating signal across keys, yielding (under standard independence assumptions) an exponentially decaying $p$-value in $n$. However, for logit–reweighting schemes that promote a per-key \emph{green} subset of size $\gamma |V|$, the intersection of promoted sets shrinks as $\gamma^n |V|$, reducing the chance that any promoted token lies in the model’s high-probability region, which makes the per-key effect attenuates and detectability saturates or even degrades.

Let $|V|$ be the vocabulary size. For secret key $\sk_i$, let $G_i\subset[|V|]$ be the green set with $|G_i|=\gamma |V|$ ($0<\gamma<1$). Under an \emph{intersection-at-generation} scheme, promotions apply to $G^{(\cap)} \;=\; \bigcap_{i=1}^n G_i,\;\text{so}\;\mathbb{E}\bigl[|G^{(\cap)}|\bigr] \approx \gamma^n |V|.$
Let $p(\cdot)$ denote the pre-boost next-token distribution at a step, and define the \emph{promoted mass} $M_n \;\coloneqq\; \sum_{v\in G^{(\cap)}} p(v).$

Since $\sum_v p(v)=1$ and each token lies in $G^{(\cap)}$ with probability $\gamma^n$, we have
$\mathbb{E}[M_n] \;=\; \gamma^n.$
If a logit boost of size $\varepsilon\leq1/\gamma$ (if $\varepsilon>1/\gamma$, then the sum of the boosted probability will be greater than 1) is applied to each promoted token during each logit-based reweighting, the per-step shift in a DiPmark-style score satisfies $\mu(n) \;\approx\; \varepsilon^n\mathbb{E}[M_n]=(\varepsilon\gamma)^n.$

\medskip\noindent\textbf{Aggregation gain vs.\ sparsity loss.}
Let $S_i$ be the per-key detector score and $S_{\ens}=\sum_{i=1}^n S_i$. Under standard independence assumptions and boundedness of scores, a Hoeffding/Chernoff bound yields
$p_{\ens}(\x_{1:n})
\;\lesssim\;
\exp\!\Big(-CTn\mu(n)^2\Big)
\;=\;
\exp\!\Big(-CTn(\varepsilon\gamma)^{2n}\Big),$
for some constant $C>0$ and sequence length $T$. This expression makes the trade-off explicit:
\[
\underbrace{n}_{\text{aggregation gain}}
\qquad\text{vs.}\qquad
\underbrace{(\varepsilon\gamma)^{2n}}_{\text{promotion sparsity}}.
\]
Define $g(n)\coloneqq n\,(\varepsilon\gamma)^{2n}$. Then $g(n)$ increases only up to
$n^\star \;\approx\; \frac{1}{2\log(1/\varepsilon\gamma)},$
and decreases thereafter. Consequently, $p_{\ens}$ typically \emph{decreases} with $n$ for $n\lesssim n^\star$ (improving detectability), but \emph{stalls} and can effectively \emph{worsen} for $n\gg n^\star$ as the promoted mass $\mathbb{E}[M_n]=\gamma^n$ becomes vanishingly small (in the extreme, $\gamma^n |V| \lesssim 1$ so no token is promoted at many steps).

\paragraph{Design implication.}
In practice, choose a \emph{moderate} ensemble size $n$ so that $n\approx n^\star$. For instance, if $\gamma=0.5,\varepsilon=1.8$, then $n^\star\approx 1/(2\log(1/0.9))\approx 4.75$, suggesting $n\in\{4,5\}$ is near optimal under strict intersection. Larger $n$ can be viable only if the ensemble design avoids strict intersection (e.g., by aggregating per-key logits or statistics) so that the per-key effect size $\mu$ does not collapse with $n$.




\begin{table}[t]
\centering
\caption{Detectability comparison of different watermarking methods under 250- and 500-token settings. 
We report True Positive Rates (TPR) at fixed False Positive Rates (FPR) of 0.1\%, 0.01\%, and 0.001\%, along with the median $p$-values (lower is better). }
\label{tab:detectability}
\resizebox{\textwidth}{!}{%
\begin{tabular}{@{}l|cccc|cccc@{}}
\toprule
\multirow{3}{*}{\begin{tabular}[c]{@{}l@{}}Watermarking\\ Methods\end{tabular}} & \multicolumn{4}{c|}{250 tokens} & \multicolumn{4}{c}{500 tokens} \\ \cmidrule(l){2-9} 
 & \multicolumn{3}{c|}{TPR@FPR} & \multirow{2}{*}{\begin{tabular}[c]{@{}c@{}}Median\\ p-value\end{tabular}$\downarrow$} & \multicolumn{3}{c|}{TPR@FPR} & \multirow{2}{*}{\begin{tabular}[c]{@{}c@{}}Median\\ p-value\end{tabular}$\downarrow$} \\ \cmidrule(lr){2-4} \cmidrule(lr){6-8}
 & 0.1\%$\uparrow$ & 0.01\%$\uparrow$ & \multicolumn{1}{c|}{0.001\%$\uparrow$} &  & 0.1\%$\uparrow$ & 0.01\%$\uparrow$ & \multicolumn{1}{c|}{0.001\%$\uparrow$} &  \\ \midrule
DiPmark($\alpha$=0.3) & 41.48\% & 32.22\% & 23.28\% & 4.48e-3 & 71.84\% & 61.68\% & 50.50\% & 8.60e-6 \\
ENS-DiPmark($\alpha$=0.3, $n$=5) & 75.05\% & 66.77\% & 58.79\% & 9.77e-7 & 96.30\% & 91.51\% & 86.62\% & 3.28e-14 \\
ENS-DiPmark($\alpha$=0.3, $n$=10) & 74.49\% & 63.21\% & 55.18\% & 2.49e-6 & 93.33\% & 87.51\% & 82.56\% & 3.75e-14 \\ \midrule
DiPmark($\alpha$=0.4) & 49.90\% & 39.18\% & 31.86\% & 1.20e-3 & 80.48\% & 70.94\% & 62.17\% & 1.27e-7 \\
ENS-DiPmark($\alpha$=0.4, $n$=5) & 74.74\% & 67.15\% & 58.52\% & 5.36e-7 & 94.18\% & 90.23\% & 85.29\% & 9.19e-15 \\
ENS-DiPmark($\alpha$=0.4, $n$=10) & 70.58\% & 58.74\% & 51.23\% & 8.19e-6 & 91.75\% & 85.99\% & 79.59\% & 1.19e-12 \\ \midrule
$\gamma$-reweight & 52.00\% & 42.02\% & 33.99\% & 7.47e-4 & 81.26\% & 72.45\% & 65.16\% & 4.58e-8 \\
ENS-$\gamma$-reweight($n$=5) & 73.98\% & 64.14\% & 54.92\% & 2.04e-6 & 94.12\% & 88.58\% & 83.59\% & 4.81e-15 \\
ENS-$\gamma$-reweight($n$=10) & 65.81\% & 55.54\% & 48.56\% & 1.75e-5 & 88.97\% & 83.35\% & 76.97\% & 2.07e-11 \\ \midrule
SynthID($m$=20) & 93.92\% & 88.55\% & 81.97\% & 6.04e-12 & 99.68\% & 98.60\% & 97.64\% & 2.83e-26 \\
SynthID($m$=30) & 93.83\% & 88.36\% & 83.50\% & 1.91e-12 & 99.24\% & 98.37\% & 97.50\% & 4.07e-28 \\
SynthID($m$=40) & 92.12\% & 87.58\% & 82.42\% & 4.79e-12 & 99.24\% & 97.84\% & 96.11\% & 7.70e-28 \\ \midrule
MCMark($l$=20) & 94.36\% & 90.37\% & 86.07\% & 4.18e-13 & 99.34\% & 98.45\% & 97.01\% & 8.30e-26 \\
ENS-MCMark($l$=20, $n$=3) & \textbf{95.70\%} & \textbf{91.71\%} & \textbf{87.93\%} & \textbf{1.43e-14} & \textbf{99.89\%} & \textbf{99.57\%} & \textbf{98.59\%} & 2.58e-31 \\
ENS-MCMark($l$=20, $n$=5) & 95.15\% & 91.44\% & 85.88\% & 4.27e-14 & 99.12\% & 98.90\% & 98.02\% & \textbf{1.27e-35} \\ \bottomrule
\end{tabular}%
}
\end{table}

\begin{table}[t]
\centering
\caption{Robustness comparison under paraphrasing attacks using GPT-4o-mini and DIPPER.}
\label{tab:paraphrase_comparison}
\resizebox{\textwidth}{!}{%
\begin{tabular}{@{}l|ccc c|ccc c@{}}
\toprule
 \multirow{3}{*}{\begin{tabular}[c]{@{}l@{}}Watermarking\\ Methods\end{tabular}} & \multicolumn{4}{c|}{GPT-4o-mini paraphrase} & \multicolumn{4}{c}{DIPPER} \\ 
\cmidrule(lr){2-5} \cmidrule(lr){6-9}
 & \multicolumn{3}{c}{TPR@FPR} & \multirow{2}{*}{\begin{tabular}[c]{@{}c@{}}Median\\ p-value$\downarrow$\end{tabular}} 
 & \multicolumn{3}{c}{TPR@FPR} & \multirow{2}{*}{\begin{tabular}[c]{@{}c@{}}Median\\ p-value$\downarrow$\end{tabular}} \\ 
\cmidrule(lr){2-4} \cmidrule(lr){6-8}
 & 0.1\%$\uparrow$ & 0.01\%$\uparrow$ & 0.001\%$\uparrow$ & & 
   0.1\%$\uparrow$ & 0.01\%$\uparrow$ & 0.001\%$\uparrow$ & \\ \midrule
ENS-DiPmark($\alpha$=0.3, $n$=5) & 8.41\% & 5.14\% & 2.80\% & 3.15e-1 & 3.26\% & 1.09\% & 0.00\% & 2.61e-1 \\
ENS-DiPmark($\alpha$=0.4, $n$=5) & 7.11\% & 3.56\% & 3.56\% & 2.13e-1 & 4.38\% & 1.03\% & 0.77\% & 3.55e-1 \\
ENS-$\gamma$-reweight($n$=5)     & 6.73\% & 3.59\% & 3.14\% & 2.78e-1 & 3.05\% & 1.39\% & 0.83\% & 3.97e-1 \\
SynthID($m$=30)                & 25.39\% & 13.47\% & 6.74\% & 1.72e-2 & 21.58\% & 11.05\% & 7.11\% & 2.06e-2 \\
ENS-MCMark($l$=20, $n$=3)        & \textbf{40.00\%} & \textbf{29.44\%} & \textbf{20.56\%} & \textbf{3.69e-3} & 
                                 \textbf{40.35\%} & \textbf{30.70\%} & \textbf{22.22\%} & \textbf{5.09e-3} \\ 
\bottomrule
\end{tabular}
}
\vspace{-0.5cm}
\end{table}

\begin{table}[t]
\centering
\caption{Robustness comparison under GPT-4o-mini back translation (English–French) attacks using GPT-4o-mini and 10\% random token replacement.}
\label{tab:back_translation_vs_random}
\resizebox{\textwidth}{!}{%
\begin{tabular}{@{}l|ccc c|ccc c@{}}
\toprule
 \multirow{3}{*}{\begin{tabular}[c]{@{}l@{}}Watermarking\\ Methods\end{tabular}} & \multicolumn{4}{c|}{Back translation (En–Fr)} & \multicolumn{4}{c}{10\% Random token replacement} \\ 
\cmidrule(lr){2-5} \cmidrule(lr){6-9}
 & \multicolumn{3}{c}{TPR@FPR} & \multirow{2}{*}{\begin{tabular}[c]{@{}c@{}}Median\\ p-value$\downarrow$\end{tabular}} 
 & \multicolumn{3}{c}{TPR@FPR} & \multirow{2}{*}{\begin{tabular}[c]{@{}c@{}}Median\\ p-value$\downarrow$\end{tabular}} \\ 
\cmidrule(lr){2-4} \cmidrule(lr){6-8}
 & 0.1\%$\uparrow$ & 0.01\%$\uparrow$ & 0.001\%$\uparrow$ & & 
   0.1\%$\uparrow$ & 0.01\%$\uparrow$ & 0.001\%$\uparrow$ & \\ \midrule
ENS-DiPmark($\alpha$=0.3, $n$=5) & 38.74\% & 26.31\% & 18.74\% & 4.61e-3 & 38.74\% & 26.31\% & 18.74\% & 4.61e-3 \\
ENS-DiPmark($\alpha$=0.4, $n$=5) & 44.38\% & 29.83\% & 20.99\% & 2.35e-3 & 44.38\% & 29.83\% & 20.99\% & 2.35e-3 \\
ENS-$\gamma$-reweight($n$=5)     & 41.71\% & 28.86\% & 21.42\% & 3.53e-3 & 41.71\% & 28.86\% & 21.42\% & 3.53e-3 \\
SynthID($m$=30)                & 75.69\% & 64.53\% & 55.58\% & 3.16e-6 & 75.69\% & 64.53\% & 55.58\% & 3.16e-6 \\
ENS-MCMark($l$=20, $n$=3)        & \textbf{84.17\%} & \textbf{76.43\%} & \textbf{68.53\%} & \textbf{1.94e-8} & 
                                 \textbf{84.17\%} & \textbf{76.43\%} & \textbf{68.53\%} & \textbf{1.94e-8} \\ 
\bottomrule
\end{tabular}
}
\vspace{-0.5cm}
\end{table}

\section{Experiments}

Our experiments comprise three main parts. First, we evaluate the detectability gains of our ensemble watermark by comparing it with other unbiased watermarking methods on a text generation task. Second, we assess robustness under random token modifications, DIPPER paraphrasing attacks~\citep{krishna2023paraphrasing}, GPT paraphrasing attacks and GPT back translation attacks. Finally, we verify the unbiasedness of our method by showing that its output quality on machine translation and text summarization tasks closely matches that of the unwatermarked baseline. All experiments are conducted on NVIDIA A6000 GPUs. All watermarking algorithms introduced negligible computation cost during LLM generation process. Detailed experimental settings are provided in Appendix~\ref{sec:detailed_experiment_setup}.

\paragraph{Baselines.} We evaluate our method against several baselines, including three logit-based unbiased watermarking algorithms: $\gamma$-reweight~\citep{hu2023unbiased}, DiPmark~\citep{wu2023dipmark}, and MCmark~\citep{chen2025improved}, as well as one sampling-based unbiased watermarking algorithm~\citep{dathathri2024scalable}. While other sampling-based unbiased methods such as ITS-edit~\citep{kuditipudi2023robust} and STA-1~\citep{mao2024watermark} exist, prior work~\citep{chen2025improved} has shown that they perform worse than MCmark, and therefore we omit them from our experiments.

\paragraph{Models and Datasets.}\
We employ Llama-3.2-3B-Instruct \citep{dubey2024llama}, Mistral-7B-Instruct-v0.3 \citep{jiang2023mistral}, and Phi-3.5-mini-instruct \citep{abdin2024phi} for text generation tasks to assess the effectiveness of our proposed \methodname. Following prior work \citep{kirchenbauer2023watermark,hu2023unbiased}, we conduct experiments on a subset of the C4 dataset \citep{raffel2020exploring}. In addition, we include evaluations on three MMW datasets \citep{piet2023mark}, Dolly CW \citep{DatabricksBlog2023DollyV2}, and two tasks from WaterBench \citep{tu2023waterbench}.

For unbiasedness validation, we adopt the settings from~\citet{hu2023unbiased,wu2023dipmark}, using MBart~\citep{liu2020multilingual} for machine translation and BART~\citep{lewis2019bart} for text summarization. In the machine translation experiments, we use the WMT16 ro-en dataset~\citep{bojar-EtAl:2016:WMT1}. For text summarization, we use the CNN/DailyMail dataset~\citep{see-etal-2017-get}.

\paragraph{Watermarking parameters.} We evaluate the detectability of \methodname\ on the text generation task with different language models. We generate ~1,000 examples for each task. We use the prefix 2-gram together with a secret key as the watermark keys. We select $\alpha \in\{ 0.3, 0.4\}$ for DiPmark~\citep{wu2023dipmark}, tournament sampling layers $m\in\{20,30,40\}$ for SynthID~\citep{dathathri2024scalable}, $l=20$ for MCmark~\citep{chen2025improved}. For $\gamma$-reweight~\citep{hu2023unbiased}, we follow the settings in the original papers. For watermark ensemble we select $n\in\{1,5,10\}$ for DiPmark and $\gamma$ reweight and $n\in\{1,3,5\}$ for MCmark. We report true positive rate under x\% theoretical guaranteed false positive rate (TPR@FPR) and the Median p-value.
\begin{table}[t]
\centering
\caption{Unbiasedness comparison of different watermarking methods on text summarization and machine translation tasks.}
\label{tab:unbiasedness}
\resizebox{\textwidth}{!}{%
\begin{tabular}{@{}l|cccc|cc@{}}
\toprule
\multirow{3}{*}{\begin{tabular}[c]{@{}l@{}}Watermarking\\ Methods\end{tabular}}  & \multicolumn{4}{c|}{Text Summarization} & \multicolumn{2}{c}{Machine Translation} \\ \cmidrule(l){2-7} 
 & ROUGE-1 & ROUGE-2 & ROUGE-L & BERTScore & BLEU & BERTScore \\ \midrule
No Watermark & 0.3768 & 0.1327 & 0.2379 & 0.3175 & 20.35 & 0.5576 \\ \midrule
DiPmark($\alpha$=0.3) & 0.3767 & 0.1325 & 0.2384 & 0.3170 & 20.44 & 0.5583 \\
ENS-DiPmark($\alpha$=0.3, $n$=5) & 0.3760 & 0.1317 & 0.2375 & 0.3163 & 20.24 & 0.5555 \\
ENS-DiPmark($\alpha$=0.3, $n$=10) & 0.3768 & 0.1328 & 0.2383 & 0.3167 & 20.21 & 0.5588 \\ \midrule
DiPmark($\alpha$=0.4) & 0.3768 & 0.1330 & 0.2385 & 0.3178 & 20.35 & 0.5559 \\
ENS-DiPmark($\alpha$=0.4, $n$=5) & 0.3768 & 0.1326 & 0.2380 & 0.3166 & 20.36 & 0.5585 \\
ENS-DiPmark($\alpha$=0.4, $n$=10) & 0.3756 & 0.1322 & 0.2379 & 0.3163 & 20.38 & 0.5590 \\ \midrule
$\gamma$-reweight & 0.3767 & 0.1320 & 0.2376 & 0.3165 & 20.54 & 0.5588 \\
ENS-$\gamma$-reweight($n$=3) & 0.3769 & 0.1331 & 0.2385 & 0.3169 & 20.23 & 0.5577 \\
ENS-$\gamma$-reweight($n$=5) & 0.3759 & 0.1321 & 0.2378 & 0.3157 & 20.45 & 0.5579 \\ \midrule
SynthID($m$=20) & 0.3761 & 0.1323 & 0.2380 & 0.3169 & 19.82 & 0.5559 \\
SynthID($m$=30) & 0.3776 & 0.1331 & 0.2387 & 0.3175 & 20.15 & 0.5582 \\
SynthID($m$=40) & 0.3774 & 0.1336 & 0.2382 & 0.3173 & 20.28 & 0.5566 \\ \midrule
ENS-MCMark($l$=20, $n$=1) & 0.3769 & 0.1329 & 0.2386 & 0.3176 & 19.83 & 0.5543 \\
ENS-MCMark($l$=20, $n$=3) & 0.3767 & 0.1325 & 0.2380 & 0.3170 & 20.43 & 0.5589 \\
ENS-MCMark($l$=20, $n$=5) & 0.3769 & 0.1333 & 0.2388 & 0.3177 & 20.19 & 0.5631 \\ \bottomrule
\end{tabular}%
}
\vspace{-0.5cm}
\end{table}
\subsection{Detectability}
The results in Table~\ref{tab:detectability} clearly demonstrate that our ensemble strategy consistently enhances the detectability of logit-based watermarking methods. For DiPmark and $\gamma$-reweight, applying the ensemble scheme substantially boosts TPR across all false positive rate thresholds and reduces median $p$-values, confirming that aggregation over multiple keys strengthens statistical power. More importantly, when comparing against strong baselines such as SynthID and MCMark, our ensemble framework achieves state-of-the-art performance. In particular, ENS-MCMark reaches the highest TPRs and the lowest $p$-values under both 250- and 500-token settings, surpassing all competing methods and establishing our ensemble method as the most effective approach for watermark detectability

\vspace{-0.1cm}
\subsection{Robustness}
\vspace{-0.1cm}
To comprehensively evaluate robustness, we conduct experiments under 4 challenging text corruption attacks: GPT-4o-mini paraphrasing, DIPPER paraphrasing, GPT-4o-mini English–French back translation, and 10\% random token replacement. These transformations substantially alter surface forms while preserving semantics, providing a rigorous stress test for watermark detectability. As shown in Tables~\ref{tab:paraphrase_comparison} and \ref{tab:back_translation_vs_random}, all watermarking methods experience degraded performance under these attacks. Nevertheless, our ensemble framework consistently yields significant improvements, with ENS-MCMark achieving the highest TPR across all FPR thresholds and the lowest $p$-values in every attack scenario. These results highlight that, even under strong paraphrasing and token-level perturbations, our ensemble method maintains state-of-the-art detectability, clearly outperforming existing baselines.

\vspace{-0.1cm}
\subsection{Unbiasedness}
\vspace{-0.1cm}
To assess the unbiasedness of watermarking, we evaluate generation quality across two representative tasks: text summarization and machine translation, using multiple standard metrics. For summarization, we report ROUGE-1/2/L and BERTScore, while for translation we adopt BLEU and BERTScore. See Appendix~\ref{sec:detailed_experiment_setup} for a detailed introduction of the metrics. As shown in Table~\ref{tab:unbiasedness}, all watermarking methods, including our ensemble variants, achieve scores that are nearly identical to the no-watermark baseline. This indicates that, similar to other unbiased watermarking approaches, our ensemble framework does not degrade generation quality. The consistency across diverse metrics and tasks confirms that the improved detectability of our ensemble method comes without sacrificing semantic fidelity or fluency of the generated outputs.

\graphicspath{{figures/}} 

\begin{figure*}[t]
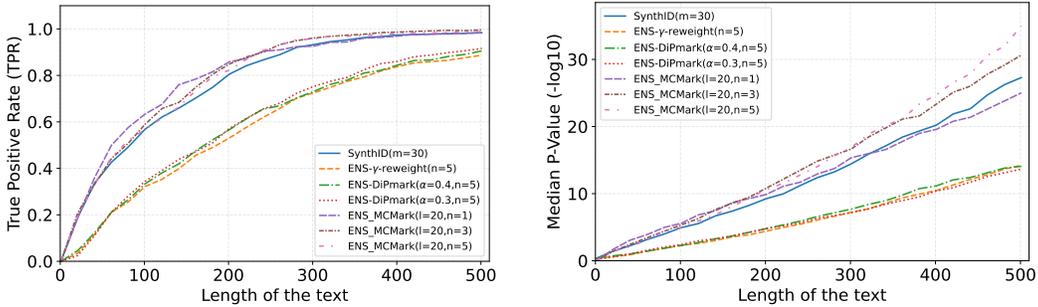

  \centering
  \begin{minipage}[t]{0.48\textwidth}\centering
    \includegraphics[width=\linewidth]{combined_acc_vs_len_Llama_3.2_3B_Instruct_c4_subset.pdf}
    
  \end{minipage}\hfill
  \begin{minipage}[t]{0.48\textwidth}\centering
    \includegraphics[width=\linewidth]{combined_median_p_vs_len_Llama_3.2_3B_Instruct_c4_subset.pdf}
  \end{minipage}
  \vspace{-0.5cm}
  \caption{Effect of the generation length on the detectability. Left: TPR@0.01\%FPR vs. generation length. Right: median $p$-value vs. generation length.}
  \label{fig:llama_pair}
\end{figure*}

\begin{figure}[t]
    \centering
    \includegraphics[width=\textwidth]{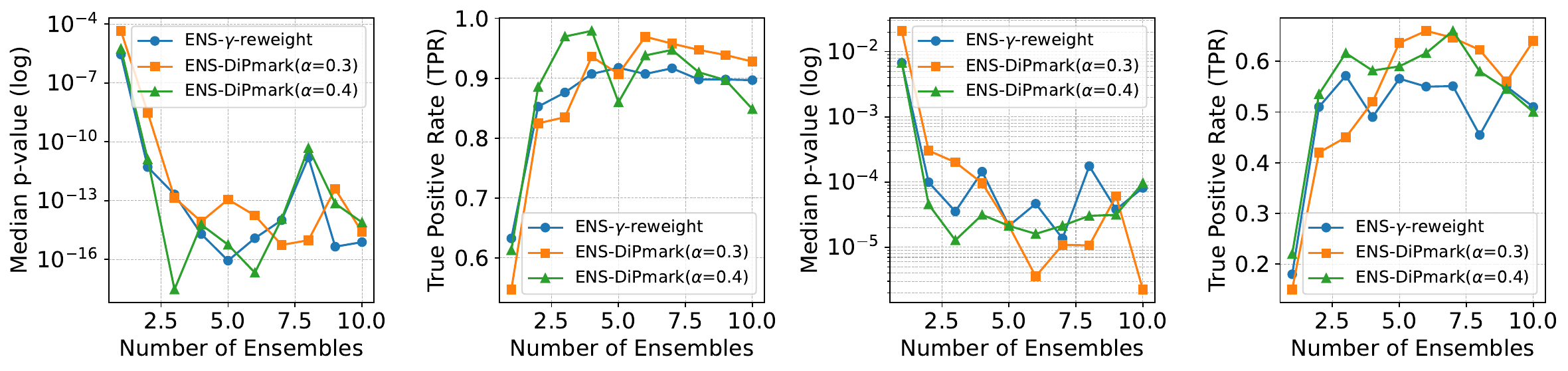}
    \vspace{-0.5cm}
    \caption{Effect of the number of ensembles on detectability. We compare $\gamma$-reweight and DiPmark ($\alpha=0.3, 0.4$) under generation lengths 250 and 500. 
    Left two plots: length=500; right two plots: length=250.}
    \label{fig:dipmark_vs_gamma_all}
    \vspace{-0.3cm}
\end{figure}

\subsection{Ablation study}
In this section, we study the effect of generation length and ensemble size on detectability. All experiments are conducted on Llama-3.2-3B-Instruct using the C4 subset, and we report TPR@0.01\% FPR together with the median $p$-value.

\paragraph{Detectability vs. generation length.}
As shown in Figure~\ref{fig:llama_pair}, increasing the generation length consistently improves detectability for all watermarking methods. Longer sequences provide more statistical evidence, thereby reducing the variance of detection scores and leading to both higher TPR and lower median $p$-values. Among the baselines, we observe that ENS-MCmark with $n=3$ achieves the best overall detectability across lengths. Compared to the original MCmark baseline (\textit{ENS-MCmark}$(\ell=20,n=1)$), our ensemble framework significantly boosts detection power while maintaining robustness. These results highlight the importance of sequence length in watermark detection and further confirm the benefit of ensemble-based designs.

\paragraph{Detectability vs. number of ensembles.}
We further analyze the effect of ensemble size using DiPmark and $\gamma$-reweight with $n=1,2,\dots,10$ under generation lengths of $250$ and $500$ (Figure~\ref{fig:dipmark_vs_gamma_all}). Interestingly, detectability does not grow monotonically with $n$: we observe that detection power initially improves as ensembles aggregate complementary evidence, but then gradually declines when $n$ becomes large. This non-monotonic trend is consistent with our theoretical analysis in Sec.~\ref{sec:ensemble size}, where excessive ensemble averaging introduces redundancy and dilutes the effective signal. In particular, moderate ensemble sizes (e.g., $n=3$–$5$) provide the best trade-off, achieving the lowest $p$-values and the highest TPR across both generation lengths.






\section{Conclusion}\label{sec:conclusion}
In this work, we introduced \methodname, a principled ensemble framework for unbiased watermarking that amplifies detection signals while rigorously preserving the underlying language model distribution. By composing multiple independent unbiased watermarks, \methodname\ achieves a provable $\sqrt{n}$ gain in signal-to-noise ratio without sacrificing imperceptibility. Our theoretical analysis confirms that the unbiasedness property holds under independent keys, and our experiments demonstrate consistent improvements in detection accuracy and robustness to text modification attacks across diverse model families and baseline methods. These results suggest that \methodname\ offers a practical and scalable path toward stronger, more reliable watermarking for real-world deployment, helping ensure provenance and accountability in the era of large-scale AI text generation.

\bibliography{iclr2026_conference}
\bibliographystyle{iclr2026_conference}

\appendix
\clearpage
\section{LLM Usage}
We ONLY used ChatGPT-4o and ChatGPT-5 to refine the content.
\section{Watermarking Algorithms}
\begin{algorithm}[h]
\caption{\methodname\ generator.}\label{alg:generator}
\begin{algorithmic}[1]
\State \textbf{Input:} pretrained LM $P_M$, secret keys $\textsf{sk}_1,\cdots,\textsf{sk}_n$, prompt $\bm{x}_{-m:0}$, generate length $T\in\mathbb{N}$, n-gram window length $a$, logits reweight strategy $F$, hash function $h$, n-gram history $hist$.
\For{$t=1,\dots,T$}
\State Initialize $P_{M,w}^{(0)}=P_M$
\If{$\x_{t-a,t-1}\in hist$}
\State Sampling from the original distribution $P_{M}(\cdot|\x_{-m:t-1})$.
\Else
\State Update $hist$ with $\x_{t-a,t-1}$.
\For{$i=1,\dots,n$}
\State Generate watermark key $k_i=h(\sk_{i},\x_{t-a,t-1})$.
\State $P_{M,w}^{(i)}(\cdot|\x_{-m:t-1}) := F(P_{M,w}^{(i-1)}(\cdot|\x_{-m:t-1})|k_i)$.
\EndFor
\State Sample the next token $x_{t}$ from $P_{M,w}^{(n)}(\cdot|\x_{-m:t-1})$.
\EndIf
\EndFor
\State \textbf{return} $\bm{x}_{1:T}$.
\end{algorithmic}
\end{algorithm}

\begin{algorithm}[h]
\caption{\methodname\ detector.}\label{alg:detector}
\begin{algorithmic}[1]
\State \textbf{Input:} pretrained LM $P_M$, secret keys $\textsf{sk}_1,\cdots,\textsf{sk}_n$, generated tokens $\bm{x}_{1:T}$, threshold $\Phi_0$, score function $s$, logits reweight strategy $F$, hash function $h$.
\State Initialize $\Phi=0$
\For{$t=1,\dots,T$}
\For{$i=1,\dots,n$}
    \State Recover the watermark key $k_i=h(\sk_{i},\x_{t-a,t-1})$.
    \State $\Phi=\Phi+s(x_t|F,k_i)$.
\EndFor
\EndFor
\If{$\Phi\geq \Phi_0$}
\State \textbf{return} $\bm{x}_{1:T}$ is watermarked.
\Else
\State \textbf{return} $\bm{x}_{1:T}$ is not watermarked.
\EndIf
\end{algorithmic}
\end{algorithm}
\section{Missing Proofs}
\begin{theorem}[Unbiasedness] \label{thm:ensemble-unbiased} If $F$ is an unbiased logits-based reweight strategy, and the watermark keys $\k_{1:n}$ are i.i.d. from $P_\cK$, then the $n$-ensemble of $F$ is also an unbiased logits-based reweight strategy, i.e., 
\begin{equation}
    \mathbb{E}_{\k_{1:n}\sim P_\cK^n}[\ens(n,F,P_M(x_{t+1}|\x_{1:t}),\k_{1:n})]=P_{M}(x_{t+1} \mid \x_{1:t})
\end{equation} holds for arbitrary $P_M(\cdot|\x_{1:t})\in\cP$ and $x_{t+1}\in V$,
\end{theorem}

\begin{proof}
    We prove it by induction, when $n=1$, since $F$ is an unbiased logit-based reweight strategy, we have \begin{equation}
    \begin{split}
    \mathbb{E}_{k\sim P_\cK}[\ens(1,F,P_M(x_{t+1}|\x_{1:t}),k)]&=\mathbb{E}_{k\sim P_\cK}[F(P_{M}(x_{t+1} \mid \x_{1:t}|k))]\\&=P_{M}(x_{t+1} \mid \x_{1:t}).
    \end{split}
    \end{equation}
    When $m>1$, assuming $$\mathbb{E}_{\k_{1:n}\sim P_\cK^n}[\ens(n,F,P_M(x_{t+1}|\x_{1:t}),\k_{1:n})]=P_{M}(x_{t+1} \mid \x_{1:t})$$ holds for $n=m-1$, when $n=m$, since $k_m$ is independent of $\k_{1:m-1}$,
    \begin{equation}\label{eqn:proof1}
    \begin{split}
    &\mathbb{E}_{\k_{1:m}\sim P_\cK^m}[\ens(m,F,P_M(x_{t+1}|\x_{1:t}),\k_{1:m})]\\
    =&\mathbb{E}_{k_m\sim P_\cK,\k_{1:m-1}\sim P_\cK^{m-1}}[F(\ens(m-1,F,P_M(\cdot|\x_{1:t}),\k_{1:m-1})|k_m)],\\
    =&\mathbb{E}_{\k_{1:m-1}\sim P_\cK^{m-1}}[\bbE_{k_m\sim P_\cK}[F(\ens(m-1,F,P_M(\cdot|\x_{1:t}),\k_{1:m-1})|k_m)]]
    \end{split}
    \end{equation}
    Since $F$ is an unbiased logits-based reweight strategy, we have
    $$ \mathbb{E}_{k_m\sim P_\cK}[F(\ens(m-1,F,P_M(\cdot|\x_{1:t}),\k_{1:m-1})|k_m)] = \ens(m-1,F,P_M(\cdot|\x_{1:t}),\k_{1:m-1}).$$
    Together with Eq.~\ref{eqn:proof1},
    \begin{equation}
    \begin{split}
    &\mathbb{E}_{\k_{1:m}\sim P_\cK^m}[\ens(m,F,P_M(x_{t+1}|\x_{1:t}),\k_{1:m})]\\
    =&\mathbb{E}_{\k_{1:m-1}\sim P_\cK^{m-1}}[\ens(m-1,F,P_M(\cdot|\x_{1:t}),\k_{1:m-1})],\\
    \overset{\text{(*)}}{=}&P_{M}(x_{t+1} \mid \x_{1:t}).
    \end{split}
    \end{equation}
    (*) refers by induction assumption. Thus, the $n$-ensemble of $F$ is also an unbiased logits-based reweight strategy.
\end{proof}

\section{Experiment Setup} \label{sec:detailed_experiment_setup}
We assess the unbiasedness properties of various watermarking models across two seq2seq tasks: text summarization and machine translation. The experiments are implemented using the Huggingface library \citep{wolf2019huggingface}, a widely adopted framework in the NLP community for model training and sharing. All evaluations are performed on 8 NVIDIA A6000 GPUs, each equipped with 48GB of memory.

\textbf{Machine Translation.} For this task, we use the WMT'14 English (En) to Romanian (Ro) dataset, which includes 1,999 test examples. The Multilingual BART (MBart) model \citep{liu2020multilingual} is adopted, along with its official tokenizer.

\textbf{Text Summarization.} For summarization, we utilize the CNN-DM test set \citep{hermann2015teaching}, consisting of 11,490 examples. We evaluate with the BART-large model (400M parameters) and the LLaMA-2 model with 7B parameters.

\textbf{Evaluation Metrics for Text Quality.} To quantify generation quality, we adopt the following metrics:
\begin{itemize}
    \item \textbf{ROUGE Score.} Applied to summarization, ROUGE \citep{lin2004rouge} measures n-gram overlap between generated summaries and reference texts, reflecting content preservation.
    \item \textbf{BLEU Score.} For translation, BLEU \citep{papineni2002bleu} evaluates lexical similarity between system outputs and human references.
    \item \textbf{BERTScore.} BERTScore \citep{zhang2019bertscore} computes semantic similarity via token embedding cosine similarity. We report BERTScore-F1, -Precision, and -Recall for both summarization and translation.
    \item \textbf{Perplexity.} Derived from information theory, perplexity measures how well a probability model predicts text. Lower values indicate stronger predictive capacity. We use it to assess both summarization and text generation.
\end{itemize}

\section{Additional Results}
In this section, we provide all the experimental results on three LMs: Llama-3.2-3B-Instruct \citep{dubey2024llama}, Mistral-7B-Instruct-v0.3 \citep{jiang2023mistral}, and Phi-3.5-mini-instruct \citep{abdin2024phi}; and 7 datasets: C4 subset \citep{raffel2020exploring}, three MMW datasets \citep{piet2023mark}, Dolly CW \citep{DatabricksBlog2023DollyV2}, and two datasets from WaterBench \citep{tu2023waterbench}.

\graphicspath{{figures/apx_results}} 

\newcommand{\pairrow}[2]{%
  \noindent
  \begin{minipage}[t]{0.49\textwidth}\centering
    \includegraphics[width=\linewidth]{combined_acc_vs_len_#1_#2.pdf}
  \end{minipage}\hfill
  \begin{minipage}[t]{0.49\textwidth}\centering
    \includegraphics[width=\linewidth]{combined_median_p_vs_len_#1_#2.pdf}
  \end{minipage}\par\vspace{0.8em}
}


\pairrow{Llama_3.2_3B_Instruct}{c4_subset}
\pairrow{Llama_3.2_3B_Instruct}{dolly_cw}
\pairrow{Llama_3.2_3B_Instruct}{finance_qa}
\pairrow{Llama_3.2_3B_Instruct}{longform_qa}
\pairrow{Llama_3.2_3B_Instruct}{mmw_book_report}
\pairrow{Llama_3.2_3B_Instruct}{mmw_fake_news}
\pairrow{Llama_3.2_3B_Instruct}{mmw_story}

\pairrow{Mistral_7B_Instruct_v0.3}{c4_subset}
\pairrow{Mistral_7B_Instruct_v0.3}{dolly_cw}
\pairrow{Mistral_7B_Instruct_v0.3}{finance_qa}
\pairrow{Mistral_7B_Instruct_v0.3}{longform_qa}
\pairrow{Mistral_7B_Instruct_v0.3}{mmw_book_report}
\pairrow{Mistral_7B_Instruct_v0.3}{mmw_fake_news}
\pairrow{Mistral_7B_Instruct_v0.3}{mmw_story}

\pairrow{Phi_3.5_mini_instruct}{c4_subset}
\pairrow{Phi_3.5_mini_instruct}{dolly_cw}
\pairrow{Phi_3.5_mini_instruct}{finance_qa}
\pairrow{Phi_3.5_mini_instruct}{longform_qa}
\pairrow{Phi_3.5_mini_instruct}{mmw_book_report}
\pairrow{Phi_3.5_mini_instruct}{mmw_fake_news}
\pairrow{Phi_3.5_mini_instruct}{mmw_story}

\end{document}